\documentclass[conference,a4paper]{IEEEtran}

\usepackage[utf8]{inputenc} 
\usepackage[T1]{fontenc}
\usepackage{url}              % provides \url{...}
\usepackage{cite}             % improves presentation of citations
\usepackage{amsthm,amsfonts}
\usepackage[cmex10]{amsmath}  % Use the [cmex10] option to ensure complicance
\usepackage{mleftright}       % fix to wrong spacing of \left-,
\mleftright                   % \middle- \right-commands 

\usepackage{graphicx}         % provides \includegraphics{...} to
\usepackage{booktabs}         % fixes poor spacing in tables and
\usepackage{tikz}
\usepgflibrary{shapes.misc}
\usetikzlibrary{arrows,positioning,fit,backgrounds}
\usetikzlibrary{shapes.misc} % LATEX and plain TEX when using Tik Z
\usetikzlibrary{calc}

\newtheorem{defn}{Definition}
\newtheorem{lem}{Lemma}
\newtheorem{thm}{Theorem}

\newtheorem{cor}[lem]{Corollary}
\newtheorem{rem}{Remark}

\DeclareMathOperator{\conv}{conv}

\begin{document}

\title{From Soft-Minoration to Information-Constrained Optimal Transport and Spiked Tensor Models} 

%%%%%%
\author{%
%  \IEEEauthorblockN{Anonymous Authors}
%  \IEEEauthorblockA{%
%    Please do NOT provide authors' names and affiliations\\
%    in the paper submitted for review, but keep this placeholder.\\
%    ISIT23 follows a \textbf{double-blind reviewing policy}.}
\IEEEauthorblockN{Jingbo Liu}
\IEEEauthorblockA{Department of Statistics, University of Illinois, Urbana-Champaign\\
jingbol@illinois.edu}
}

\maketitle

%%%%%
%% Abstract: 
%% If your paper is eligible for the student paper award, please add
%% the comment "THIS PAPER IS ELIGIBLE FOR THE STUDENT PAPER
%% AWARD." as a first line in the abstract. 
%% For the final version of the accepted paper, please do not forget
%% to remove this comment!
%%
\begin{abstract}
Let $P_Z$ be a given distribution on $\mathbb{R}^n$.
For any $y\in\mathbb{R}^n$, we may interpret $\rho(y):=\ln\mathbb{E}[e^{\left<y,Z\right>}]$ as a soft-max of $\left<y,Z\right>$. We explore lower bounds on $\mathbb{E}[\rho(Y)]$ in terms of the minimum mutual information $I(Z,\bar{Z})$ over $P_{Z\bar{Z}}$ which is a coupling of $P_Z$ and itself such that $Z-\bar{Z}$ is bounded in a certain sense.  This may be viewed as a soft version of Sudakov's minoration, which lower bounds the expected supremum of a stochastic process in terms of the packing number.
Our method is based on convex geometry (thrifty approximation of convex bodies), and works for general non-Gaussian $Y$.
When $Y$ is Gaussian and $\bar{Z}$ converges to $Z$,
this recovers a recent inequality of  Bai-Wu-Ozgur on information-constrained optimal transport, previously established using Gaussian-specific techniques. 
We also use soft-minoration to  obtain asymptotically (in tensor order) tight bounds on the free energy in the Sherrington-Kirkpatrick model with spins uniformly distributed on a type class, implying asymptotically tight bounds for the type~II error exponent in spiked tensor detection.
\end{abstract}

\section{Introduction}
Given $P_Z$ on $\mathbb{R}^n$, define
$
\rho(y):=\ln\mathbb{E}
[e^{\left<y,Z\right>}],
\quad y\in\mathbb{R}^n$,
where $Z\sim P_Z$ and $\left<,\right>$ denotes the inner product. 
We may interpret $\rho(y)$ as a soft-max of $\left<y,Z\right>$.
Indeed, if $P_Z$ is the uniform distribution on a compact set $\mathcal{A}$, then $\rho(y)\le \max_{z\in\mathcal{A}}\left<y,z\right>$.
Moreover, the inequality typically becomes tight when $y$ is large.
If $P_Y$ is the standard Gaussian distribution, Sudakov's minoration \cite{sudakov1969gaussian}\cite{ledoux2013probability} gives
\begin{align}
\mathbb{E}[\max_{z\in\mathcal{A}}\left<Y,z\right>]\ge c\sup_{l>0}\sqrt{\ln{\sf P}_l(\mathcal{A})}
\label{e2}
\end{align}
where $c>0$ is a universal constant and ${\sf P}_l(\mathcal{A})$ denotes the $l$-packing number of $\mathcal{A}$ under the Euclidean distance.
Generalization of Sudakov's minoration to other log-concave measures $P_Y$ has also been considered \cite{ledoux2013probability}\cite{latala2014sudakov}\cite{mendelson2019generalized}.
In this paper we explore inequalities of the following form which may be called ``soft minoration'':
\begin{align}
\mathbb{E}[\rho(Y)]\ge \textrm{function of }\inf I(Z;\bar{Z})
\label{e3}
\end{align}
where the inf is over coupling $P_{Z\bar{Z}}$ under which $Z-\bar{Z}$ is ``small'' and both $Z$ and $\bar{Z}$ have the same law $P_Z$.

One motivation for \eqref{e3} is network information theory. Cover's problem asks the minimum relay rate needed for achieving the maximum capacity of a relay channel \cite{cover} (see also \cite{WuBarnesOzgur}).
Measure concentration and reverse hypercontractivity  techniques yield nontrivial bounds but are not sufficient for solving Cover's problem 
\cite{wu2017improving}\cite{liu2019new}\cite{liu2020capacity}.
The solution in the Gaussian setting is infinity, as shown by  \cite{WuBarnesOzgur} using a rearrangement inequality for the spheres (see also a solution for binary symmetric channels using a similar idea \cite{BarnesWuOzgur-BSC}). 
Bai, Wu, and Ozgur \cite{bai} provided a simplified proof for the Gaussian setting by proving a bound on information constrained optimal transport: if $P_Y$ is the standard normal distribution in $\mathbb{R}^n$, $P_Z$ has well-defined differential entropy, and $R>0$, then
\begin{align}
\frac{n}{\tau_0(R/n)}\exp\left(\tfrac{h(Z)-{\sf h}_1}{n}\right)
\le
\sup_{\substack{\text{$P_{YZ}\in\Pi(P_Y,P_Z)$}\\
\text{$I(Y;Z)\le R$}}}\mathbb{E}[\left<Y,Z\right>]
\label{e4}
\end{align}
where ${\sf h}_1:=\frac{n}{2}\ln(2\pi e)$ (all information in  nats throughout the paper),
the function $\tau_0(\theta):=\frac1{\sqrt{1-e^{-2\theta}}}$, and $\Pi(P_Y,P_Z)$ denotes the set of  couplings between two distribution. 
\eqref{e4} generalizes Talagrand's $T_2$ inequality by replacing optimal transport with entropy-regularized optimal transport \cite{bai}.
Previous proofs of \eqref{e4}  relied on Gaussian specific arguments.
In contrast, \cite{el2022strengthened} used the traditional auxiliary random variable approach, yielding the same capacity region outer bound for the Gaussian setting as \cite{WuBarnesOzgur}\cite{bai},
and also showed that compress-and-forward solves Cover's problem for discrete memoryless channels under a full-rankness condition.
Concurrently,
\cite{liu2021soft}\cite{liu2022minoration} used convex geometry to lower bound $\mathbb{E}[\rho(Y)]$ with packing numbers of $\mathcal{A}$ (for $P_Z$ uniform on $\mathcal{A}$), showing the optimality of compress-and-forward for all discrete memoryless channels under conditions originally stated in \cite{cover} (without full-rankness assumption).

In \cite{liu2022minoration} the argument was restricted to $P_Z$ of the form of a uniform distribution on a set, which is sufficient for the solution of Cover's problem because when applied to the relay channel, $P_Z$ is the restriction of the channel output distribution on the intersection of a type class and a relay decoding set.
In this paper, we consider general $P_Z$, and the packing number of a set is replaced by the mutual information in \eqref{e3},
which is accomplished by combining the approach of \cite{liu2022minoration} with a tensorization argument.
Another difference from \cite{liu2022minoration} is that \cite{liu2022minoration} used the reduction of general $P_Y$ to the case of Rademacher distribution;
in contrast, this paper uses Barvinok's thrifty approximation of convex bodies \cite{barvinok2014thrifty} which entails explicit information theoretic bounds for general $P_Y$.

We show as a consequence of \eqref{e3} that
for any $l>0$,
\begin{align}
&\quad\inf_{\substack{\text{$P_{Z\bar{Z}}\in\Pi(P_Z,P_Z)$}\\
\text{$\mathbb{E}[\|Z-\bar{Z}\|^2]\le \frac1{4}\tau_0^2(\frac{R}{n})l^2n$}}}I(Z;\bar{Z})
\nonumber
\\
&\le 
n\ln\left(
1+\frac{2}{nl}\sup_{\substack{\text{$P_{YZ}\in\Pi(P_Y,P_Z)$}\\
\text{$I(Y;Z)\le R$}}}\mathbb{E}[YZ]
\right)
\label{e_4}
\end{align}
which implies 
\eqref{e4} as $l\to0$. 
As noted in \cite{cuturi2013sinkhorn}\cite{bai}, information constrained optimal transport is useful in machine learning due to the availability of fast algorithms.
In many such applications $P_Z$ is the empirical distribution of samples, in which case $h(Z)=-\infty$ and \eqref{e4} is useless. 
In contrast, a bound with $I(Z;\bar{Z})$ may still be nontrivial.
Moreover, our approach easily extends to general symmetric distribution $P_Y$, where $\tau_0(\cdot)$ is replaced by another universal function $\tau(\cdot)$ and ${\sf h}_1$ or the infimum in \eqref{e_4} is given a general definition.

In statistical physics, a central quantity is the (expected) free energy $\mathbb{E}[\ln\int e^{H(\sigma)}d\mu(\sigma)]$,
where $\sigma$ is called the spin or configuration, and the expectation is with respect to the randomness (also known as the disorder) of the Hamiltonian $H(\cdot)$ \cite{talagrand2010mean}\cite{mezard2009information}.
Clearly \eqref{e3} provides a lower bound on the free energy, once we embed $H$ and $\sigma$ in a suitable Euclidean space so that $H(\sigma)$ is an inner product.
The free energy in the Sherrington-Kirkpatrick model (SK) characterizes the information-theoretic threshold for spiked tensor detection \cite{chen2019phase}.
Statistical physics literature has been focusing on the cases of Rademacher and spherical spins, and existing exact formulae for the free energy generally rely on these structures and are usually hard to evaluate \cite{talagrand2006free}\cite{subag2017geometry}\cite{chen2019phase}\cite{mourrat2022parisi}.
We prove a simple dimension-free bound in the format of \eqref{e3} and show its tightness when the tensor order is large and the prior is uniform on a type class, 
which in turn provides asymptotically tight type II error exponent bounds in spiked tensor detection.

\section{Preliminaries}
The simplest proof of Sudakov's minoration \eqref{e2} in the case of Gaussian $P_Y$ is through Gaussian comparison (see e.g.\ \cite{chatterjee2005error}\cite{ledoux2013probability}).
However, this approach is Gaussian specific, and a longstanding goal in this research area is to extend the results to general log-concave measures (see \cite{mendelson2019generalized}).
For our proofs in Section~\ref{sec_general}-\ref{sec_gaussian}, it suffices to use the following result of Pajor \cite{pajort},
which can be viewed as a generalization of \eqref{e2} to the case of general $P_Y$ and packing number under the Minkowski functional distance. 
The proof is a simple application of the Alexandrov-Fenchel inequality, and a review of related convex geometry concepts can be found in \cite{liu2022minoration}.
\begin{lem}\cite{pajort}\label{lem_pajor}
Suppose that $\mathcal{C}$ is a symmetric convex body in $\mathbb{R}^N$, and let $P_Y$ be the associated cone volume measure.
Let $\mathcal{A}\subseteq\mathbb{R}^N$ be compact, and define $a:=\mathbb{E}[\sup_{z\in\mathcal{A}}\sup\left<z,Y\right>]$.
Let $\mathcal{C}^{\circ}$ be the polar of $\mathcal{C}$.
For any $l>0$, define ${\sf P}_l(\mathcal{A})$ as the $l$-packing number of $\mathcal{A}$ under the Minkowski function $\|\|_{\mathcal{C}^{\circ}}$ (which is a norm in the case of symmetric convex body $\mathcal{C}$).
Then
\begin{align}
{\sf P}_l(\mathcal{A})\le \left(1+2a/l\right)^N.
\label{e5}
\end{align}
\end{lem}
\begin{rem}\label{rem_dred}
A difference between \eqref{e2} and \eqref{e5} is that the latter is dimension dependent.
It is possible however to use a Johnson-Lindenstrauss embedding argument to reduce $N$ in \eqref{e5} to the order of $\ln{\sf P}_l(\mathcal{A})$ and recover \eqref{e2}; see \cite{mendelson2019generalized}.
\end{rem}

Another key ingredient for the proofs in Section~\ref{sec_general} is to relate $\mathbb{E}[\rho(Y)]$ in \eqref{e3} to 
$\mathbb{E}[\max_{z\in\mathcal{A}}\left<Y,z\right>]$,
which is achieved by approximating the support of $P_Y$ in \eqref{e2} by a sparser set and then apply Markov's inequality and the union bound. 
More specifically, we use the following ``thrifty approximation of convex body'' by Barvinok \cite{barvinok2014thrifty}. 
We state here a simplified asymptotic version.
\begin{lem}\label{lem_barvinok}\cite{barvinok2014thrifty}
For any $\tau>1$, let $\kappa>0$ and $\theta>0$ be the solutions to 
\begin{align}
\frac{1+\kappa}{2\kappa}
h(\tfrac{\kappa}{1+\kappa})
&=\ln(\tau+\sqrt{\tau^2-1});
\label{e6}
\\
(1+\kappa)h(\tfrac{\kappa}{1+\kappa})&=\theta,
\label{e7}
\end{align}
where $h(\cdot)$ denotes the binary entropy function.
Then for any symmetric convex body $\mathcal{C}\subseteq \mathbb{R}^N$, 
there exists a symmetric polytope $P$ satisfying $P\subseteq C\subseteq \tau P$ and with at most $e^{\theta N+o(N)}$ vertices.
\end{lem}
\begin{rem}
From \eqref{e6}\eqref{e7}, if $\tau\to1$ then $\kappa=\frac{1+o(1)}{4\sqrt{2(\tau-1)}}\ln\frac1{\tau-1}$ and $\theta=\frac{1+o(1)}{2}\ln\frac1{\tau-1}$.
If $\tau\to\infty$ then $\kappa=\frac{1+o(1)}{\tau^2}$ and $\theta=\frac{2+o(1)}{\tau^2}\ln \tau$. 
%in the special case of hypercube,  provided an explicit construction which yields an improved estimate of $\tau=O(\frac1{\tau^2})$.
\end{rem}

\section{General $P_Y$}\label{sec_general}
In this section we derive a bound in the form of \eqref{e3} which, among other things, generalizes \eqref{e4} to arbitrary $P_Y$ satisfying $P_Y=P_{-Y}$ (see Corollary~\ref{cor4}).
To be precise, $\tau_0(R/n)$ in \eqref{e4} will be replaced by a worse constant for general $P_Y$; we will explain in the next section how the constant is improved to $\tau_0(R/n)$ for Gaussian $P_Y$.

Let $\theta(\tau)$ be the function defined implicitly in \eqref{e6}-\eqref{e7}. 
\begin{thm}\label{thm_main}
Suppose that $P_Y$ and $P_Z$ are distributions on $\mathbb{R}^n$, $\mathbb{E}[Y]=0$, and $Y$ and $-Y$ have the same distribution, where $Y\sim P_Y$.
Then for any $l>0$ and $\tau>1$,
\begin{align}
\inf_{P_{\bar{Z}Z}}I(\bar{Z};Z)
\le n\ln\left(1+\frac{2}{ln}\left(\mathbb{E}[\rho(Y)]+n\theta(\tau)\right)\right)
\label{e22}
\end{align}
where the infimum is over all $P_{\bar{Z}Z}\in \Pi(P_{Z},P_Z)$ satisfying $\mathbb{E}[\left<\bar{Z}-Z,Y\right>]\le \frac{\tau ln}{2}$ for all $P_{\bar{Z}ZY}\in\Pi(P_{\bar{Z}Z},P_Y)$.
\eqref{e22} also holds if $\mathbb{E}[\rho(Y)]$ is replaced by $\sup_{P_{YZ}\in\Pi(P_Y,P_Z)}
\left\{-I(Y;Z)+\mathbb{E}[\left<Y,Z\right>]
\right\}$.
\end{thm}
\begin{proof}
It suffices to prove the case where $P_Y$ and $P_Z$ are supported on  finite sets and all the probability masses are rational numbers. 
The general case can then be established by an approximation argument, using the fact that the mutual information can be arbitrarily well approximated with finite partitions of the space \cite{pinsker}.
For any $N>0$ which divides the denominators of these rational numbers, 
%define 
%\begin{align}
%\mathcal{L}^{\perp}:=\left\{(x_1,\dots,x_N)\colon x_1=\dots=x_N\in\mathbb{R}^n
%\right\}
%\end{align}
%which is an $n$-dimensional subspace of $\mathbb{R}^N$, and let $\mathcal{L}$ be its orthogonal complement.
%Let $P_{\tilde{Y}^N}$ be the equiprobable distribution on the $P_Y$ type class.
%Define the weakly typical set $\mathcal{A}$ as
%\begin{align}
%\left\{z_1,\dots,z_N\in\mathbb{R}^n\colon
%|\sum_{i=1}^N \ln \frac1{P_Z(z_i)}-Nh(P_Z)|\le N^{2/3}
%\right\},
%\end{align}
let $P_{\tilde{Y}^N}$ (resp.\ $P_{\tilde{Z}^N}$) be the equiprobable distribution on $\mathcal{C}$ (resp.\ $\mathcal{A}$), defined as the $P_Y$-type class (resp.\ $P_Z$-type class).
For any $y^N\in\mathbb{R}^{nN}$, define 
\begin{align}
\tilde{\rho}(y^N):=
\ln\mathbb{E}[e^{\left<y^N,\tilde{Z}^N\right>}]
\end{align}
where $\tilde{Z}^N\sim P_{\tilde{Z}^N}$.
Then by the method of types and large deviation analysis we have
\begin{align}
\mathbb{E}[\tilde{\rho}(\tilde{Y}^N)]
&=N\sup_{P_{YZ}\in\Pi(P_Y,P_Z)}
\left\{-I(Y;Z)+\mathbb{E}[\left<Y,Z\right>]
\right\}
+o(N)
\label{e25}
\\
&\le N\mathbb{E}[\rho(Y)]+o(N)
\label{e_25}
\end{align}
where $(Y,Z)\sim P_{YZ}$ in \eqref{e25},
and \eqref{e_25} follows since 
by the Donsker-Varadhan variational formula,  
$\rho(y)
=
\sup_{Q_Z}\left\{\mathbb{E}_{Q_Z}[\left<y,Z\right>]-D(Q_Z\|P_Z)
\right\}$ for any $y$ and therefore
\begin{align}
\mathbb{E}[\rho(Y)]\ge
\sup_{P_{YZ}\in\Pi(P_Y,P_Z)}
\left\{-I(Y;Z)+\mathbb{E}[\left<Y,Z\right>]
\right\}.\label{e26}
\end{align}
By Lemma~\ref{lem_barvinok}, we can choose $\mathcal{S}$ as a subset of the convex hull of $\mathcal{C}$ such that 
\begin{align}
\ln|\mathcal{S}|&=nN\theta(\tau)+o(N);
\label{e13}
\\
\mathcal{S}^{\circ}&\subseteq \tau\mathcal{C}^{\circ}.
\label{e27}
\end{align}
Let $\hat{Y}^N$ be equiprobable on $\mathcal{S}$.
Define $\mathcal{B}$ a subset of $\mathbb{R}^{nN}$ as
\begin{align}
\bigcap_{y^N\in\mathcal{S}}\{z\in\mathcal{A}\colon \left<y^N,z^N\right>\le
\mathbb{E}[\left<y^N,\tilde{Z}^N\right>]
+\tilde{\rho}(y^N)+\ln(2|\mathcal{S}|)\}.
\end{align}
Then by Markov's inequality we have 
\begin{align}
P_{\tilde{Z}^N}[\mathcal{B}]
\ge \frac1{2},
\label{e29}
\end{align}
and moreover,
\begin{align}
\mathbb{E}[\sup_{z^N\in\mathcal{B}}
\left<\hat{Y}^N,z^N\right>]
&\le \mathbb{E}[\tilde{\rho}(\hat{Y}^N)]
+\ln(2|\mathcal{S}|)
\\
&\le 
\mathbb{E}[\tilde{\rho}(\tilde{Y}^N)]
+\ln(2|\mathcal{S}|)
\label{e31}
\end{align}
where we used $\mathbb{E}[\hat{Y}^N]=0$,
and the fact that $\tilde{\rho}(\cdot)$ is a constant on $\mathcal{C}$ by permutation invariance of the type class.
Now let ${\sf P}_{\tau nNl}(\mathcal{B})$ be the $\tau nNl$-packing number under $\|\|_{\mathcal{C}^{\circ}}$,
which is upper bounded by the $nNl$-packing number $\|\|_{\mathcal{S}^{\circ}}$ by \eqref{e27}.
Therefore by Pajor Lemma~\ref{lem_pajor},
\begin{align}
\ln{\sf P}_{\tau nNl}(\mathcal{B})
&\le 
nN\ln\left(1+\frac{2}{nNl}\mathbb{E}\left[\sup_{z^N\in\mathcal{B}}
\left<\hat{Y}^N,z^N\right>\right]\right).
\label{e32}
\end{align}
For any $z^N\in\mathcal{A}$, the set $(z^N+\frac{\tau nNl}{2}\mathcal{C}^{\circ})\cap \mathcal{A}$ is 
\begin{align}
\left\{\bar{z}^N\in\mathcal{A}\colon \left<\bar{z}^N-z^N,y^N
\right>\le \frac{\tau nNl}{2},\textrm{if } y^N \textrm{ is $P_Y$-type}
\right\},
\label{e33}
\end{align}
whose $\ln$ cardinality is, by large deviation analysis,
\begin{align}
N\sup_{P_{\bar{Z}Z}}H(\bar{Z}|Z)
+o(N),
\label{e34}
\end{align}
where the supremum is over the same set as the infimum in \eqref{e22}.
Note that the packing number can be lower bounded by $|\mathcal{B}|$ divided by the cardinality of the set in \eqref{e33};
using 
\eqref{e29} and \eqref{e34} we have
\begin{align}
\ln{\sf P}_{\tau nNl}(\mathcal{B})
\ge N\inf_{P_{\bar{Z}Z}}I(\bar{Z};Z)
+o(N).\label{e35}
\end{align}
The theorem follows by  \eqref{e_25}\eqref{e31}\eqref{e32}\eqref{e35} and taking $N$ large.
\end{proof}
Next, we consider a limiting case of Theorem~\ref{thm_main} as $l\to0$.
\begin{defn}\label{defn1}
Fix $P_Y$ a distribution on $\mathbb{R}^n$.
For any $L>0$ define 
\begin{align}
{\sf h}_L=\sup_{P_X}h(X).\end{align}
where the supremum is over all $P_X$ satisfying $\sup_{P_{XY}\in\Pi(P_X,P_Y)}\mathbb{E}[\left<X,Y\right>]\le nL$.
\end{defn}

\begin{rem}\label{rem1}
Since $h(LX)=h(X)+n\ln L$, we see that ${\sf h}_L={\sf h}_1+n\ln L$.
Moreover, if $Y$ is standard Gaussian then the supremum is achieved when $X=LY$ by Talagrand's inequality (special case of \eqref{e4} when $R\to\infty$), and therefore ${\sf h}_L=\frac{n}{2}\ln(2\pi eL^2)$.
\end{rem}

\begin{cor}\label{cor3}
Suppose that $P_Y$ and $P_Z$ are as in Theorem~\ref{thm_main}, and additionally, $P_Z$ has well defined differential entropy.
Let ${\sf h}_1$ be as in Definition~\ref{defn1}. Then
\begin{align}
h(Z)
\le {\sf h}_1+n\inf_{\tau>1}\ln\left(
\frac{\tau}{n}\mathbb{E}[\rho(Y)]+\tau\theta(\tau)
\right).
\end{align}
\end{cor}
\begin{proof}
Using Remark~\ref{rem1} we have
\begin{align}
I(\bar{Z};Z)
&=h(Z)-h(Z-\bar{Z}|\bar{Z})
\\
&\ge h(Z)-h(Z-\bar{Z})
\\
&\ge h(Z)-{\sf h}_{\frac{\tau ln}{2}}
\\
&=h(Z)-{\sf h}_1-n\ln \frac{\tau ln}{2}.
\end{align}
Then corollary follows by taking $l\to 0$ in \eqref{e22}.
\end{proof}

\begin{rem}\label{rem2}
From \eqref{e25}-\eqref{e_25} we see that the bounds in Theorem~\ref{thm_main} and Corollary~\ref{cor3}
can be sharpened by replacing $\mathbb{E}[\rho(Y)]$ by 
the right side of \eqref{e26}.
It might appear that this is a strict improvement, but actually it is just an equivalent version since the converse implication is also true.
Indeed, the fact that these bounds with $\mathbb{E}[\rho(Y)]$ holding for all $P_Z$ and $n$ implies the sharpened versions with the right side of \eqref{e26}, using a similar tensorization argument as in \eqref{e25}-\eqref{e_25}.
\end{rem}

An equivalent form of \eqref{cor3} is the following:
\begin{cor}\label{cor4}
Let $P_Y$ and $P_Z$ be as in Corollary~\ref{cor3}.
Denote $\tau(\cdot)$ as the inverse function of $\theta(\cdot)$. For any $R$, we have
\begin{align}
\sup_{P_{YZ}}\mathbb{E}[\left<Y,Z\right>]
\ge 
\frac{n}{\tau(R/n)}\exp\left(\frac{h(Z)-{\sf h}_1}{n}\right)
\end{align}
where the supremum is over $P_{YZ}\in\Pi(P_Y,P_Z)$ satisfying $I(Y;Z)\le R$.
%Equivalently,
%\begin{align}
%W_2^2(P_Y,P_Z;R)\le \mathbb{E}[\|Y\|^2+\|Z\|^2]-\tfrac{2n}{\tau(R/n)}\exp\left(\tfrac{h(Z)-{\sf h}_1}{n}\right).
%\end{align}
\end{cor}

\begin{proof}
Using Corollary~\ref{cor3} and Remark~\ref{rem2} we obtain
\begin{align}
&\exp(\tfrac{h(Z)-{\sf h}_1}{n})
\le \nonumber
\\
&\inf_{\theta>0}\left(
\tfrac{\tau(\theta)}{n}
\sup_{P_{YZ}\in\Pi(P_Y,P_Z)}
\left\{-I(Y;Z)+\mathbb{E}[\left<Y,Z\right>]\right\}
+\tau(\theta)\theta
\right).
\label{e43}
\end{align}
Let $\lambda>0$ be such that the $P_{YZ}$ achieving
\begin{align}
\sup_{P_{YZ}\in\Pi(P_Y,P_Z)}
\left\{-I(Y;Z)+\lambda\mathbb{E}[\left<Y,Z\right>]\right\}
\label{e44}
\end{align}
ensures that $I(Y;Z)=R$.
Make the substitution $Z\leftarrow \lambda Z$ in \eqref{e43},
and let $P_{YZ}$ be a optimal coupling for \eqref{e44}.
We have
\begin{align}
&\quad \lambda\exp(\tfrac{h( Z)-{\sf h}_1}{n})
 \nonumber
\\
&=\exp(\tfrac{h(\lambda Z)-{\sf h}_1}{n})
\\
&\le\inf_{\theta>0}\left(
\tfrac{\tau(\theta)}{n}
\left\{-I(Y;Z)+\lambda\mathbb{E}[\left<Y,Z\right>]\right\}
+\tau(\theta)\theta
\right)
\\
&\le \frac{\tau(\frac1{n}I(Y;Z))}{n}\cdot \lambda \mathbb{E}[\left<Y,Z\right>]
\end{align}
which establishes the claim.
\end{proof}

\section{The Gaussian Case}\label{sec_gaussian}
For Gaussian $P_Y$, we can improve the estimates in Section~\ref{sec_general} by replacing Lemma~\ref{lem_barvinok} with the following sharp estimate,
which follows from sphere covering (e.g.\ \cite{boroczky2003covering}).
\begin{lem}\label{lem5}
Fix $\phi\in(0,\frac{\pi}{2})$. 
For any positive integer $N$ there exists a set $\mathcal{S}_N$ on the unit ball $B_1^N$ satisfying $|\mathcal{S}_N|=\frac1{\sin^N(\phi+o(1))}$
and $\frac1{\cos\phi}\conv(\mathcal{S}_N)\supseteq B_1^N$.
\end{lem}
Now define the functions $\theta_0(\tau)$ and $\tau_0(\theta)$ by the equations
$\tau=\frac1{\cos\phi}$ and
$\theta=\ln\frac1{\sin\phi}$; explicitly,
\begin{align}
\theta_0(\tau)&=\ln\sqrt{1-\tau^{-2}};
\\
\tau_0(\theta)&=\frac1{\sqrt{1-e^{-2\theta}}}
\end{align}
\begin{thm}\label{thm2}
If $P_Y$ is the standard Gaussian distribution on $\mathbb{R}^n$, Then the bounds in Theorem~\ref{thm_main},
Corollary~\ref{cor3}, and Corollary~\ref{cor4} can be improved by replacing $\theta(\tau)$ and $\tau(\theta)$ with $\theta_0(\tau)$ and $\tau_0(\theta)$.
Moreover the left side of \eqref{e22} can be improved to $\inf_{P_{\bar{Z}Z}\colon \|Z-\bar{Z}\|^2\le\tau^2l^2n/4}I(\bar{Z};Z)$.
\end{thm}
\begin{proof}
The proof is similar to the general non-Gaussian case,
and we shall only mention a few differences in the argument.
It suffices to consider $P_Z$ supported on a finite set with all probability masses equal to rational numbers.
Let $N>0$ divide the denominators of these rational numbers. 
Define $\mathcal{A}$, $P_{\tilde{Z}^N}$, and $\tilde{\rho}$ as in the proof of Theorem~\ref{thm_main}.
Then we have
\begin{align}
h(\tilde{Z}^N)
&=Nh(Z)+o(N),
\\
\mathbb{E}[\tilde{\rho}(Y^N)]
&=
\mathbb{E}\left[\ln\int_{\mathcal{A}} e^{\left<Y^N,z^N\right>}d {P_Z}^{\otimes N}(z^N)\right]+o(N)
\nonumber\\
&\le \mathbb{E}\left[\ln\int e^{\left<Y^N,z^N\right>}d {P_Z}^{\otimes N}(z^N)\right]+o(N)
\nonumber\\
&=N\mathbb{E}[\rho(G)]+o(N).
\end{align}
Define $r_N:=\mathbb{E}[\|Y^N\|]$ and let $\tilde{Y}^N:=\frac{r_N}{\|Y^N\|}Y^N$.
Note that $\mathbb{E}[Y^N|\tilde{Y}^N]=\tilde{Y}^N$, and therefore by Jensen's inequality,
\begin{align}
\mathbb{E}[\tilde{\rho}(\tilde{Y}^N)]\le 
\mathbb{E}[\tilde{\rho}(Y^N)].
\end{align}
Choose $\mathcal{S}$ similar to before but use Lemma~\ref{lem5} instead and replace $\theta(\tau)$ in \eqref{e13} with $\theta_0(\tau)$.
Define $\hat{Y}^N$ as the random variable distributed on $\mathcal{S}$ and following the cone volume measure, and let $U$ be a random rotation in $\mathbb{R}^{nN}$, independent of $\hat{Y}^N$ and following the uniform distribution on the orthogonal group.
Then
$
\mathbb{E}[\tilde{\rho}(\tilde{Y}^N)]=
\mathbb{E}[\tilde{\rho}(U\hat{Y}^N)$.
There exists some (deterministic) rotation $u$ such that $\mathbb{E}[\tilde{\rho}(U\hat{Y}^N)]\ge \mathbb{E}[\tilde{\rho}(u\hat{Y}^N)]$, which we can assume without loss of generality to be the identity, so that 
\begin{align}
\mathbb{E}[\tilde{\rho}(\tilde{Y}^N)]
\ge \mathbb{E}[\tilde{\rho}(\hat{Y}^N)].
\end{align}
There rest of the proof is similar to Theorem~\ref{thm_main}, where $\mathcal{C}$ is now the centered sphere of radius $r_N=\sqrt{N}(1+o(1))$.
The improved estimate on the left side of \eqref{e22} is seen by refining \eqref{e33} for $y^N$ in a ball.
\end{proof}
\begin{rem}
The bounds claimed in Theorem~\ref{thm2} are  asymptotically tight (as $n\to\infty$) when $P_Z$ is uniform on a ball.
\end{rem}

\section{Spiked Tensor Model}
In this section we explore bounds in the form \eqref{e3} when $Z=X^{\otimes d}$ has a special rank-1 tensor structure, and the implications for the spiked tensor detection problem in statistics \cite{perry2018optimality}\cite{chen2019phase}.
The order $d$ tensors associated with $R^n$ is again a vector space, and can be given an inner product compatible with the Frobenius norm.
The dimension of order $d$ tensors is $n^d$ and order $d$ symmetric tensors is ${n+d-1}\choose{d}$, both too large for directly applying a dimension depending minoration such as Lemma~\ref{lem_pajor} for tight bounds.
As mentioned in Remark~\ref{rem_dred}, a dimension reduction argument may be applied. 
In this section, we shall just focus on the case of Gaussian $P_Y$ where we can apply a Gaussian comparison argument,
which reduces to the random energy model (REM).
The result we will use is
(see \cite[p150]{mezard2009information})
\begin{align}
\lim_{M\to\infty}\frac1{M}\mathbb{E}
\left[
\ln\sum_{j=1}^{2^M}
e^{-\beta E_j}
\right]
=\left\{
\begin{array}{cc}
\frac{\beta^2}{4}+\ln2 &(\beta <2\sqrt{\ln 2})
\\
\beta\sqrt{\ln 2} &(\beta\ge 2\sqrt{\ln 2})
\end{array}
\right.
\label{e_rem}
\end{align}
where $E_j\sim \mathcal{N}(0,M/2)$, $j=1,\dots,2^M$ are independent.

We will lower bound the soft-max (free energy) when $X\in\mathbb{R}^n$  follows the equiprobable distribution on a type class;
once this setting is understood, the free energy for other permutation invariant $P_X$ (such as i.i.d.\ coordinates) can be estimated using standard method of types and large deviation analysis.
\begin{thm}\label{thm4}
Let $P_{\sf X}$ be a  distribution on $\mathbb{R}$ with finite support, and $\sqrt{n}X=\sqrt{n}({\sf X}_1,\dots,{\sf X}_n)$ be equiprobably distributed on the $P_{\sf X}$-type class (with rounding if necessary).
Define $Z=\sqrt{\frac{n}{2}}\lambda X^{\otimes d}$ and 
\begin{align}
\rho(y):=\ln\mathbb{E}[e^{\left<y,Z\right>}]
\end{align}
for any order $d$ tensor $y\in\mathbb{R}^{n^d}$,
and let
\begin{align}
I_{\epsilon}:=\inf_{\substack{\text{$P_{\sf X\bar{X}}\in\Pi(P_{\sf X},P_{\sf X})$}\\
\text{$\mathbb{E}^d[{\sf X\bar{X}}]\ge \mathbb{E}^d[|{\sf X}|^2]-\epsilon^2/2$}}}I({\sf X;\bar{X}}).
\label{e_je}
\end{align}
Then 
\begin{align}
\lim_{n\to\infty}\frac1{n}\mathbb{E}[\rho(G)]
\ge
\sup_{\epsilon>0}\left\{
\begin{array}{cc}
\frac{\lambda^2\epsilon^2}{8}&
(\lambda^2\epsilon^2<8I_{\epsilon})\\
\lambda\epsilon\sqrt{\frac1{2}I_{\epsilon}}-I_{\epsilon}
&(\lambda^2\epsilon^2\ge 8I_{\epsilon})
\end{array}
\right.
\end{align}
where $G\in\mathbb{R}^{n^d}$ is an order $d$ tensor with i.i.d.\ standard Gaussian entries.
\end{thm}
\begin{proof}
Set $r=\sqrt{\frac{n}{2}}\lambda\epsilon$, where $\epsilon>0$ does not depend on $n$.
Define
\begin{align}
k_{\epsilon,n}:=
\mathbb{E}^{-1}[P_{Z}
(Z+B(r))],
\end{align}
where $B(r)$ denotes the centered ball of radius $r$ in the space of tensors under the Frobenius norm.
Let $\sqrt{n}x$ and $\sqrt{n}\bar{x}$ be two sequences in the $P_{\sf X}$-type class, and define $z$ and $\bar{z}$ accordingly.
Let $t$ be the joint type of $(x,\bar{x})$.
Then
\begin{align}
\|z-\bar{z}\|^2
&=\tfrac{n\lambda^2}{2}
\|x^{\otimes d}-\bar{x}^{\otimes d}\|^2
=n\lambda^2(\mathbb{E}_{P_{\sf X}}^d[{\sf X}^2]-\mathbb{E}_t^d[{\sf X\bar{X}}]),
\label{e51}
\end{align}
where $\mathbb{E}_t$ denotes the expectation under the type $t$, viewed as a distribution on $\mathbb{R}^2$.
Then by the large deviations analysis,
\begin{align}
\lim_{n\to\infty}\frac1{n}\ln k_{\epsilon,n}
&=I_{\epsilon},
\end{align}
We now generate a random measure $\hat{\nu}$ supported on $\mathcal{A}$, the support of $P_Z$. 
Select uniformly at random a point in $\mathcal{A}$, 
and then select uniformly at random the next point among all points in $\mathcal{A}$ at least $r$ away from the previously selected points, and so on, until no more points can be selected.
Let $\hat{\nu}$ be the equiprobable distribution on these selected points. 
$\hat{\nu}$ is random because of the randomness of the point selection process.
By symmetry of the type class, we see that $\mathbb{E}[\hat{\nu}]=P_Z$, so by Jensen's inequality,
\begin{align}
\mathbb{E}[\rho(G)]\ge 
\mathbb{E}\left[\ln\int e^{\left<G,z\right>}d\hat{\nu}(z)\right].
\label{e_43}
\end{align}
Since the support size of $\mu$ is at least $k_{r,N}$,
using \eqref{e_rem} and Slepian's comparison \cite{chatterjee2005error}, we can lower bound the right side of \eqref{e_43} in terms of the free energy of the REM with parameters $M,\beta$ given by 
\begin{align}
M&=\log_2 k_{r,N};\\
M\beta^2&=r^2
\end{align}
and the theorem follows by taking $n\to\infty$.
\end{proof}
While the statistical physics literature mostly focuses on $X$ equiprobable on a Boolean cube, a general $X$ is relevant for statistical applications such as spiked tensor detection \cite{perry2018optimality}\cite{chen2019phase}.
Let the noise $W\in\mathbb{R}^{n^d}$ be a tensor with  i.i.d.\ $\mathcal{N}(0,\frac{2}{n})$ entries.
Consider a hypothesis testing problem with observation 
\begin{itemize}
\item $H_0$: $T=W$;
\item $H_1$: $T=\lambda X^{\otimes d}+W$,
\end{itemize}
where $\lambda>0$ is the signal to noise ratio.
Denote by $P_{H_0}$ and $P_{H_1}$ the distributions of $T$ under $H_0$ and $H_1$, respectively.
From the Gaussian density formula it is easy to see that 
\begin{align}
D(P_{H_0}\|P_{H_1})=\frac{n\lambda^2}{4}
-\mathbb{E}\left[\ln\mathbb{E}[e^{\frac{n\lambda}{2} \left<W,X^{\otimes d}\right>}|W]\right].
\label{e_35}
\end{align}
Using concentration, it can be shown that the critical $\lambda$ for detecting a rank-1 spike with nontrivial probability coincides with the largest $\lambda$ for 
$D(P_{H_0}\|P_{H_1})=o(n)$.
Previously \cite{perry2018optimality} computed such critical $\lambda$ by bounding $D(P_{H_0}\|P_{H_1})$ with the R\'enyi divergence $D_2(P_{H_1}\|P_{H_0})$.
However when $\lambda$ is above the critical value, $D_2(P_{H_1}\|P_{H_0})$ grows super-linearly in $n$ (see \cite[Section~2.4]{perry2018optimality}) and hence does not give a useful bound for  $D(P_{H_0}\|P_{H_1})$ and hence for the free energy.
In contrast, we show that Theorem~\ref{thm4} is asymptotically tight for large $d$:
\begin{cor}\label{cor_4}
In Theorem~\ref{thm4}, suppose that $P_{\sf X}$ has unit variance.
Then
\begin{align}
\lim_{d\to\infty}\lim_{n\to\infty}\frac1{n}\mathbb{E}[\rho(G)]
=
\left\{
\begin{array}{cc}
\frac{\lambda^2}{4}&
(\lambda<2\sqrt{H})\\
\lambda\sqrt{H}-H
&(\lambda\ge 2\sqrt{H})
\end{array}
\right.
\label{e47}
\end{align}
where $H$ is the entropy of $P_{\sf X}$.
In particular, if $\lambda\ge 2\sqrt{H}$ and the type~I error in spiked tensor detection is bounded away from 0 and 1, then the optimal type~II exponent converges to $(\lambda/2-\sqrt{H})^2$ as $d\to\infty$.
\end{cor}
\begin{proof}
For any $\epsilon\in (0,\sqrt{2})$, $I_{\epsilon}$ in Theorem~\ref{thm4} converges to $H$ as $d\to\infty$. Taking $\epsilon\uparrow\sqrt{2}$ proves the $\ge$ part. 
To see the $\le$ part, we follow \cite{perry2018optimality} and consider the maximum likelihood statistic
\begin{align}
m:=\max_v\left<T,v^{\otimes d}\right>
\end{align}
where the maximum is over $v$ such that $\sqrt{n}v$ is $P_{\sf X}$-typical.
Let $\mathcal{E}_n$ be the event that $m\le m_n$, where $m_n$ is defined as the number such that $P_{H_0}(\mathcal{E}_n)=\frac1{2}$.
By the union bound calculation in \cite[Proposition 4.1]{perry2018optimality}, we have $\lim_{n\to\infty}m_n\le 2\sqrt{H}$.
Then 
\begin{align}
P_{H_1}(\mathcal{E}_n)
&\le P_{H_1}(\left<T,X^{\otimes d}\right>\le m_n)
\\
&=P_{H_1}(\lambda+\left<W,X^{\otimes d}\right>\le m_n).
\end{align}
Note that $\left<W,X^{\otimes d}\right>$ follows $\mathcal{N}(0,\frac{2}{n})$.
If $\lambda^2\le 2\sqrt{H}$, we have
\begin{align}
\lim_{n\to\infty}\frac1{n}\ln P_{H_1}(\mathcal{E}_n)\le -(\tfrac{\lambda}{2}-\sqrt{H})^2.
\end{align}
By the data processing inequality,
\begin{align}
\lim_{n\to\infty}\frac1{n}D(P_{H_0}\|P_{H_1})
&\ge \lim_{n\to\infty}\frac1{n}d(P_{H_0}(\mathcal{E}_n)\|P_{H_1}(\mathcal{E}_n))
\\
&\ge (\tfrac{\lambda}{2}-\sqrt{H})^2
\end{align}
where $d(\cdot\|\cdot)$ denotes the binary divergence function. Then from \eqref{e_35} we have shown the $\le$ part of \eqref{e47} in the case of $\lambda\ge2\sqrt{H}$. 
The $\le$ part in the case of $\lambda<2\sqrt{H}$ is trivial from $D(P_{H_0}\|P_{H_1})\ge 0$.
\end{proof}
\begin{rem}
Results related to Corollary~6 have appeared in the literature: as mentioned,
\cite{perry2018optimality} performed 2-R\'enyi divergence calculations to show that the critical $\lambda$ converges to $2\sqrt{H}$ as $d\to\infty$.
The 2-R\'enyi divergence is equivalent to the expected partition function of 2-replica systems.
For $X$ equiprobable on the hypercube, a classical replica symmetry calculation (see e.g.\ \cite{mezard2009information}) shows that the free energy of the $d$-spin model converges to the free energy of the REM as $d\to\infty$.
\end{rem}

\section{acknowledgement}
The author would like to thank Qiang Wu for discussions on the SK model and tensor detection.

%
%\begin{thm}
%In Theorem~\ref{thm4}, if $X$ is instead uniformly distributed on the unit sphere, then the claim still holds with $P_{\sf X}$ set to be the standard normal distribution.
%Moreover
%\begin{align}
%I_{\epsilon}=\frac1{2}\ln\frac1{1-\rho^2},
%\end{align}
%where $\rho\in(0,1)$ is the solution to $\rho^d=1-\frac{\epsilon^2}{2}$.
%\end{thm}
%
%\begin{rem}
%The result is consistent with, which showed that the critical inverse temperature $\lambda_{\sf crit}$ is $\sqrt{\ln d}+o(1)$ for large $d$.
%Suppose that $\limsup_{n\to\infty}(\lambda_d^2-2\ln\frac{d}{2}+2\ln\ln\ln d)< 0$.
%Then 
%\begin{align}
%\lim_{d\to\infty}(\lambda_d^2/4-\lim_{n\to\infty}\frac1{n}\mathbb{E}[\rho(G)])=0.
%\label{e48}
%\end{align}
%Choose $\epsilon_d$ according to
%\begin{align}
%1-\frac{\epsilon_d^2}{2}=\frac1{\ln d\cdot \ln\ln d}.
%\end{align}
%Then
%\begin{align}
%\lambda_d^2\epsilon^2-8I_{\epsilon}
%&=
%2\lambda_d^2-\frac{2\lambda_d^2}{\ln d\cdot\ln\ln d}
%-4\ln\frac1{1-(\ln d\cdot\ln\ln d)^{-\frac{2}{d}}}
%\\
%&=2\lambda_d^2+o(1)-4\ln\frac{d}{2(\ln\ln d+\ln\ln\ln d)}
%\\
%&=2\lambda_d^2-4\ln\frac{d}{2}+4\ln\ln\ln d+o(1)
%\end{align}
%which is negative for sufficiently large $d$.
%Therefore the left side of \eqref{e48} is upper bounded by
%\begin{align}
%\frac{\lambda_d^2}{4}(1-\frac{\epsilon_d^2}{2})
%=\frac{\lambda_d^2}{4\ln d\cdot \ln\ln d}
%\end{align}
%which vanishes as $d\to\infty$, and the claim is established.
%\end{rem}
%

\bibliographystyle{plainurl}
\bibliography{ref_minorization}
\end{document}